\documentclass[11pt]{article}
\usepackage{times}
\usepackage{amsmath,amsthm, amssymb}
\usepackage[ruled,vlined]{algorithm2e}
\usepackage{fullpage}
\usepackage{color}
\usepackage{tikz}
\usetikzlibrary{arrows.meta,positioning}
\usepackage{algorithmic}
\usepackage{setspace}
\usepackage{enumitem}
\usepackage{wrapfig}
\usepackage{subcaption}
\usepackage{titling}
\usepackage{authblk}
\usepackage{mathrsfs} 

\usepackage{pstricks,pst-node}
\usepackage{comment} 

\newtheorem{theorem}{Theorem} 
\newtheorem{corollary}[theorem]{Corollary}
\newtheorem{lemma}{Lemma}

\newtheorem{claim}{Claim}
\newtheorem{definition}{Definition}

\newtheorem*{conjecture*}{Conjecture}
\newtheoremstyle{nonindented}{1ex}{1ex}{}{}{\bfseries}{.}{.5em}{}
\newtheoremstyle{indented}{1ex}{1ex}{\itshape\addtolength{\leftskip}{0.6cm}\addtolength{\rightskip}{0.6cm}}{}{\bfseries}{.}{.5em}{}
\theoremstyle{nonindented}
\theoremstyle{indented}

\theoremstyle{plain}

\renewcommand{\tilde}{\widetilde}
\renewcommand{\bar}{\overline}

\DeclareMathOperator*{\argmax}{arg\,max}

\def\pr{\qopname\relax n{Pr}}

\def\min{\qopname\relax n{min}}
\def\max{\qopname\relax n{max}}

\def\argmax{\qopname\relax n{argmax}}

\newenvironment{lp*}{\begin{equation*}  \begin{array}{lll}}{\end{array}\end{equation*}}

\title{Better Approximation for Interdependent SOS Valuations} 

\author{
Pinyan Lu$^1$~~~
Enze Sun$^2$~~~
Chenghan Zhou$^3$~~~
}
\affil{\small $^1$ ITCS, Shanghai University of Finance and Economics \\
$^2$  Department of Computer Science, The University of Hong Kong \\
$^3$ Department of Computer Science, Princeton University}
\date{\vspace{-6ex}}

\begin{document}

\maketitle

\begin{abstract}
Submodular over signal (SOS) defines a family of interesting functions for which there exist truthful mechanisms with constant approximation to the social welfare for agents with interdependent valuations. The best-known truthful auction is of $4$-approximation and a lower bound of $2$ was proved. We propose a new and simple truthful mechanism to achieve an approximation ratio of $3.315$.  
\end{abstract}

\section{Introduction}

In most study of auction theory, it is assumed that valuations are agents’ private information and they know their own values when they submit their bids to the auctioneer. However, this is not usually the case in real life. For example, when one buys a house or an art work by auction, his valuation largely depends on other’s valuations since they will impact the item’s resale value later on. When an advertiser bids an impression or click in internet, the value largely depends on that particular customer for which other bidders may have more information. To describe the valuation interdependence between different bidders, a model is proposed by Milgrom and Weber~\cite{milgrom1982competitive}. Each bidder $i$ holds some private information about the item, denoted by a signal $s_i\in \mathbb{R}^+$. Agent $i$'s valuation when receiving the item is a public-known function $v_i(\mathbf{s})$ that depends on the signals of all bidders. This has become the standard model for interdependent value settings (IDV) and has been studied in the economics literature for a few decades~\cite{Chawla2014approximate, Eden2018withoutsinglecrossing,Eden2019SOS,Dasgupta2000efficient,Roughgarden2016optimal,Eden2021poa,Takayuki2006contingent,Gkatzelis2021priorfree}.

For this model without any restriction on the valuation functions, it is impossible to design a truthful auction\footnote{The truthfulness notion here is ex-post IC \& IR rather than DSIC since DSIC is not possible for interdependent valuation} with good social welfare guarantee~\cite{Jehiel2001efficient}. This is in strong contrasts to the private valuation model, for which the VCG mechanism can achieve truthfulness and optimal welfare  simultaneously~\cite{Vickrey1961,Clarke1971,Groves1973,McLean2015implementation}. A natural extension of VCG mechanism only works when the valuation functions satisfy a technical condition called single-crossing condition~\cite{Ausubel2000vickrey,Athey2001singlecrossing,Vohra2011MechanismDA,Li2013approximation,Che2015EfficientAW}. However, there are many relevant settings where the single-crossing condition does not hold~\cite{Eden2018withoutsinglecrossing,Maskin1992privatization,Dasgupta2000efficient}. 

A different and beautiful perspective is proposed by Eden et al~\cite{Eden2019SOS}.  They introduced a new condition of the valuation functions called submodular over signal (SOS) property. Submodular captures a natural diminishing returns property, which is very common in economics settings. They designed a simple random sampling auction to achieve an approximation of 4 and proved that no truthful auction can do better than 2-approximation. 

This 4-approximation remains the state-of-art for general SOS settings.  The only improvement was made for the very special case of binary signal, where the signal for each agent only has two possible values. For this special binary signal setting, it is proved that there exists a tight 2-approximation auction~\cite{Amer2021ImprovedSOS}. It is an existential proof rather than an efficient design. To construct such a 2-approximation auction may need exponential time. 

\subsection*{Our contributions} 

Firstly, we generalize the random sampling auction in \cite{Eden2019SOS}. In \cite{Eden2019SOS}, they evenly divided the agents into two sets  and their analysis paired a set of agents with its complement to prove their approximation ratio. We generalize this to sampling with arbitrary probability $p$ and get a similar approximation ratio in terms of $p$.  Although the result is as expected, the proof is completely new since the original pairing trick does not work for biased sampling. We also get a more careful analysis of the approximation with a term involving the ratio of the largest and second largest values.  This part is easy but crucial to the final improvement of the mechanism. The observation is quite intuitive: the random sampling mechanism performs much better when the second largest value is comparable to the largest one. 

Secondly, we proposed a brand new auction. Our mechanism is very simple, it allocates the item to agent $i$ with the probability $$\frac{1}{2}\left(1-\frac{\max_{j \in [n], j \neq i} v_j(\mathbf{s}_{-i}, 0)}{\max_{j \in [n]} v_j(\mathbf{s})}\right).$$
We call our mechanism the contribution-based mechanism. The intuition and the meaning of the name will be discussed in Section 3.  
This mechanism is simple, efficient to implement and easy to verify the truthfulness since the allocation rule is monotone. However, the tricky part is to verify that it is indeed a well-defined mechanism, namely the overall allocation probability cannot exceed one. This proof crucially uses the property of SOS. In particular, we obtain a lemma from the SOS property which is the key of the proof. This lemma may be of independent interests. For example, the lemma is used in the analysis of the random sampling auction.  This contribution-based mechanism’s approximation ratio is at least $\frac{1}{2} \left( 1-\frac{v_{(2)}(\mathbf{s})}{v_{(1)}(\mathbf{s})} \right)$, where $v_{(1)}(\mathbf{s})$ and  $v_{(2)}(\mathbf{s})$ are the largest and the second largest values respectively given the signal profile $\mathbf{s}$. From this expression we can see that it achieves good ratio when the largest value is much larger than others. This is in the opposite direction with the random sampling mechanism. 

Finally, we run a convex combination of the above two mechanisms. Since the random sampling mechanism performs well when the second largest value is comparable to the largest one while  contribution-based mechanism performs well when the largest one is much larger than all other values, their combination achieves a good balance for all instances. The approximation of our final mechanism is $3.315$. This improves the previous 4-approximation mechanism for the first time. 

Besides the new auction, we also investigate the relation with SOS and strong-SOS, a stronger notion of SOS which was also introduced in~\cite{Eden2019SOS}. We build a reduction and prove that strong-SOS is as difficult as SOS in terms of approximation ratio for single item setting. 
This means that it is fine to design mechanisms for strong-SOS valuation only if it is easier since the mechanism can be transformed to a mechanism for general SOS valuations with almost same approximation ratio.  In \cite{Eden2019SOS}, a better approximation ratio was given for strong-SOS when the size of signal space is restricted. This does not contradict to our result since our reduction will enlarge the signal space greatly. 

\subsection*{Related works}
In this paper, we only focus on the canonical single item setting. The original paper~\cite{Eden2019SOS} studied SOS valuations in a much broader combinatorial auction setting. Their $4$-approximation works for any single-parameter downward-closed setting with single-dimensional signals and SOS valuations. They also studied multi-dimensional signal with separable SOS valuations and gave a $4$-approximation. We defined an extend version of SOS called $d$-SOS. We are not going to define all these extended versions but focus on single item and SOS valuation for simplicity. Interested readers can find these extensions in paper~\cite{Eden2019SOS}. The above mentioned $2$-approximation auction~\cite{Amer2021ImprovedSOS} can extend to systems with matroid constraints. 
A recent paper also studied single item setting but with private SOS valuations~\cite{Eden22private}. In our setting, the signals are private while the valuation functions are public.

\section{Preliminaries}

We consider a single-item auction with $n$ bidders. In the interdependent setting, each bidder $i$ holds some private information about the item, denoted by a signal $s_i\in \mathbb{R}^+$. The signals of all bidders participating in the auction can be collected as a vector $\mathbf{s} = (s_1, s_2, \cdots, s_n)$. We sometimes use $(\mathbf{s}_{-i}, s_i)$ to emphasize the signal of agent $i$. The signal space is denoted by $\mathcal{S}$. By convention, we assume that $s_i=0$ is the minimum signal in each agent $i$'s signal space.  

Agent $i$'s valuation when receiving the item is a public-known function $v_i(\mathbf{s})$ that depends on the signals of all bidders. By convention, we assume that $v_i(\mathbf{s})$ is non-negative, weakly increasing in all signals and strongly increasing in bidder $i$'s signal. 
$v_{(k)}(\mathbf{s})$ denotes the $k$th largest valuation of a single agent when the signals of all bidders is $\mathbf{s}$. 
We focus on the case where the valuation function for each agent is \emph{submodular over signals (SOS)}:

\begin{definition}[Submodularity over signals]
     A valuation function $v(\mathbf{s})$ is \emph{submodular over signals} if for all bidders $i$, $\mathbf{s}'_{-i} \succeq \mathbf{s}_{-i}$ and $s'_i\geq s_i$,
     $$v(\mathbf{s}'_{-i}, s'_i) - v(\mathbf{s}'_{-i}, s_i) \leq v(\mathbf{s}_{-i}, s'_i) - v(\mathbf{s}_{-i}, s_i).$$
\end{definition}

A mechanism $M=(\mathbf{x}, \mathbf{p})$ 
decides the allocation rule $\mathbf{x}$ and payment $\mathbf{p}$. Without loss of generality, we consider direct mechanisms where bidders report their private signals $\tilde{\mathbf{s}}$ as bids. The mechanism then allocates the item to bidder $i$ with probability $x_i(\tilde{\mathbf{s}})$ and asks for payment $p_i(\tilde{\mathbf{s}})$. $\mathbf{x}$ satisfies feasibility constraint $\sum_{i \in [n]} x_i(\mathbf{s}) \leq 1$ for all signal profiles $\mathbf{s}$.

Throughout our analysis, we adopt the solution concepts of \emph{ex-post IC \& IR} mechanisms, defined as follows:
\begin{definition}
     An ex-post incentive compatible (IC) mechanism means that each bidder does not regret reporting his private signal $s_i$ truthfully after knowing all the other bidders' reported signals $\mathbf{s}_{-i}$. Formally, let the signal profile be $\mathbf{s} = (s_i, \mathbf{s}_{-i})$. For all signals $s'_i$,
     $$x_i(\mathbf{s}_{-i}, s_i) v_i(\mathbf{s}) - p_i(\mathbf{s}_{-i}, s_i) \geq x_i(\mathbf{s}_{-i}, s'_i) v_i(\mathbf{s}) - p_i(\mathbf{s}_{-i}, s'_i)$$
     
     An ex-post individual rational (IR) mechanism satisfies 
     $$x_i(\mathbf{s}_{-i}, s_i) v_i(\mathbf{s}) - p_i(\mathbf{s}_{-i}, s_i) \geq 0$$
\end{definition}

In this paper, we mainly focus on the allocation rule since our goal is to maximize social welfare. The following characterization allows us to design allocation rule alone with an additional monotone constraint. The payment can be deviated from the allocation rule by the standard method and will be omitted in this paper~\cite{Roughgarden2016optimal}. 

\begin{lemma}
For an allocation rule $\mathbf{x}$, there exists a payment rule $\mathbf{p}$ to make the mechanism $M=(\mathbf{x}, \mathbf{p})$  ex-post IC \& IR iff it satisfies the following monotonicity: for any bidder $i$, $\mathbf{s}_{-i}$  and $s'_i>s_i$, we have  $x_i(\mathbf{s}_{-i}, s'_i)\ge x_i(\mathbf{s}_{-i}, s_i)$.
    
\end{lemma}

\section{Contribution-Based Mechanism}
 Let agent $i^*$ be the agent with the maximum valuation at signal $\mathbf{s}$. 
 The optimal social welfare is $v_{i^*}(s)$. However, we do not view that this social welfare is contributed by the agent $i^*$ alone since it also depends on other agents' signals. We view the contribution of agent $i\in[n]$ ( including $i^*$) as 
 $v_{i^*}(s)-\max_{j \in [n], j \neq i} v_j(\mathbf{s}_{-i}, 0)$, where $\max_{j \in [N], j \neq i} v_j(\mathbf{s}_{-i}, 0)$ is the optimal social welfare when agent $i$ is not in the game (his signal is "zeroed out" and his valuation is excluded). This difference of social welfare is the contribution brought by agent $i$ to the game. We want   to allocate the item to the agents proportional to their contributions. That is why we call our mechanism Contribution-Based Mechanism. The ideal allocation probability for agent $i$ should be 
  $\frac{v_{i^*}(s)-\max_{j \in [n], j \neq i} v_j(\mathbf{s}_{-i}, 0)}{v_{i^*}(s)}$. Unfortunately, this is not a valid mechanism since the total probability may exceed $1$. However, we are able to prove that the total probability will never exceed $2$ due to the SOS property of the functions. Therefor, we can half the probability and get a valid mechanism. 
  
\noindent \textbf{Contribution-Based Mechanism:}
  Let agent $i^*$ be the agent with maximum valuation at signal $\mathbf{s}$.
  For every agent $i$, we allocate the item to him with the probability 
  \[ x_i(\mathbf{s})= \frac{v_{i^*}(s)-\max_{j \in [n], j \neq i} v_j(\mathbf{s}_{-i}, 0)}{2 v_{i^*}(s)} \]

Before we prove that it is indeed a valid mechanism and analyse its approximation, we first prove an important property of the SOS functions. (This is also discovered in~\cite{Eden22private}.)

\begin{lemma}
\label{lemma:prev-value-deviation}
    Let $T \subseteq [n]$ be a subset of bidders. Signals $\mathbf{s}$ and $\mathbf{s}'$ satisfy $\forall t \in T, s'_t \leq s_t$ and $\mathbf{s}_{-T} = \mathbf{s}'_{-T}$. For each bidder $i \in [n]$, 
    $$ \sum_{t \in T} \left( v_i(\mathbf{s}) - v_i(\mathbf{s}_{-t}, s'_t) \right) \leq v_i(\mathbf{s}) $$
\end{lemma}

\begin{proof}
Denote $T = \{ t_1, t_2, \cdots, t_{|T|}\}$.

\begin{align*}
    &\sum_{j = 1}^{|T|} \left( v_i(\mathbf{s}) - v_i(\mathbf{s}_{-t_j}, s'_{t_j}) \right) \\
    \leq& \sum_{j=1}^{|T|} v_i \left( \mathbf{s}_{-\{t_1, \cdots, t_{j-1} \}}, s'_{t_1}, \cdots, s'_{t_{j-1}} \right) - v_i \left( \mathbf{s}_{-\{t_1, \cdots, t_{j} \}}, s'_{t_1}, \cdots, s'_{t_{j}} \right) \\
    =& v_i(\mathbf{s}) - v_i(\mathbf{s}') \\
    \leq& v_i(\mathbf{s})
\end{align*}

The first inequality comes from the SOS property.

\end{proof}


\begin{theorem}
    Contribution-Based Mechanism is an ex-post IC \& IR mechanism with $\frac{1}{2} \left( 1-\frac{v_{(2)}(\mathbf{s})}{v_{(1)}(\mathbf{s})} \right)$-approximation, where $v_{(1)}(\mathbf{s})$ and  $v_{(2)}(\mathbf{s})$ are the largest and the second largest values respectively given the signal profile $\mathbf{s}$ .
\end{theorem}

\begin{proof}
\textbf{Correctness:} We prove that $\sum_{i \in [n]} x_i(\mathbf{s}) \leq 1$. It is clear that $x_{i^*}(\mathbf{s}) \leq 1/2$. So we only need to  show that $\sum_{i \in [n], i \neq i^*} x_i(\mathbf{s}) \leq 1/2$. 

\begin{align*}
    \sum_{i \in [n], i \neq i^*} x_i(\mathbf{s}) 
    =& \sum_{i \in [n], i \neq i^*} \frac{v_{i^*}(\mathbf{s}) - \max_{j \in [n], j \neq i} v_j(\mathbf{s}_{-i}, 0)}{2 v_{i^*}(\mathbf{s})} \\
    \leq& \sum_{i \in [n], i \neq i^*} \frac{v_{i^*}(\mathbf{s}) - v_{i^*}(\mathbf{s}_{-i}, 0)}{2 v_{i^*}(\mathbf{s})} \\
    =& \frac{\sum_{i \in [n], i \neq i^*} (v_{i^*}(\mathbf{s}) - v_{i^*}(\mathbf{s}_{-i}, 0))} {2 v_{i^*}(\mathbf{s})} \\
    \leq& \frac{v_{i^*}(\mathbf{s})}{2v_{i^*}(\mathbf{s})} \\
    =& \frac{1}{2} 
\end{align*}

The second inequality comes from Lemma \ref{lemma:prev-value-deviation}.

\textbf{Monotonicity:} For any agent $i$, any $\mathbf{s}_{-i}$ and two signals $s'_i \geq s_i$, by the weak monotonicity of valuation functions, we have 
\[\max_{j \in [n]} v_j(\mathbf{s}_{-i}, s'_i) \geq  \max_{j \in [n]} v_j(\mathbf{s}_{-i}, s_i).  \]
As a result
\begin{align*}
x_i(\mathbf{s}_{-i}, s'_i) & = \frac{\max_{j \in [n]} v_j(\mathbf{s}_{-i}, s'_i)-\max_{j \in [n], j \neq i} v_j(\mathbf{s}_{-i}, 0)}{2 \max_{j \in [n]} v_j(\mathbf{s}_{-i}, s'_i)}\\
&\geq  \frac{\max_{j \in [n]} v_j(\mathbf{s}_{-i}, s_i)-\max_{j \in [n], j \neq i} v_j(\mathbf{s}_{-i}, 0)}{2 \max_{j \in [n]} v_j(\mathbf{s}_{-i}, s_i)}  \\
&=x_i(\mathbf{s}_{-i}, s_i)
\end{align*}

Therefore, our mechanism always assigns a higher probability  to agent $i$ when its signal is stronger. There exists a payment rule to make it ex-post IC and IR. 

\textbf{Approximation:} We simply verify that the agent $i^*$ get the item with the probability at least   $\frac{1}{2} \left( 1-\frac{v_{(2)}(\mathbf{s})}{v_{(1)}(\mathbf{s})} \right)$ since $v_{i^*}(\mathbf{s})$ is the targeted optimal social welfare. 
By definition, we have $v_{(1)}(\mathbf{s})=v_{i^*}(\mathbf{s})$ and 
\[v_{(2)}(\mathbf{s}) = \max_{j \in [n], j \neq i^*} v_j(\mathbf{s}) \geq \max_{j \in [n], j \neq i^*} v_j(\mathbf{s}_{-i^*}, 0).\]
Therefore, 

\[
    x_{i^*}(\mathbf{s}) 
    = \frac{1}{2} \left( 1 - \frac{\max_{i \in [n], i \neq i^*} v_i(\mathbf{s}_{-i^*}, 0)}{v_{i^*}(\mathbf{s})} \right) \geq \frac{1}{2} \left( 1 - \frac{v_{(2)}(\mathbf{s})}{v_{(1)}(\mathbf{s})} \right). 
\]

\end{proof}


\section{Random sampling Auction}
Our random sampling auction is a generalization of the auction in \cite{Eden2019SOS}. Their auction is a special case of ours by choosing $p=0.5$. 

\noindent \textbf{Random Sampling Mechanism} 
\begin{itemize}
    \item Each agent $i$ is allocated to set $A$ with probability $p$ and set $B$ with probability $1-p$.
    \item For $i \in B$, let $w_i = v_i(\mathbf{s}_A, s_i, \mathbf{0}_{B \backslash \{i\}})$.
    \item Allocate to bidder $\argmax_{i \in B}  w_i $
\end{itemize}

It is clear that this is an ex-post IC \& IR mechanism since the allocation rule is monotone.  
In \cite{Eden2019SOS}, they proved that the approximation ratio of the Random Sampling Mechanism is $1/4$ when $p=1/2$. We shall prove that the ratio is $p(1-p)$ for general $p$ and further refine the
ratio in terms of $\frac{v_{(2)}(\mathbf{s})}{v_{(1)}(\mathbf{s})}$.

\begin{theorem}\label{thm:random}
    For every signal $\mathbf{s}$, Random Sampling Mechanism is an ex-post IC \& IR mechanism with the approximation ratio of
    $$p(1-p)\left(1 + p \cdot \frac{v_{(2)}(\mathbf{s})}{v_{(1)}(\mathbf{s})}\right).$$
\end{theorem}

We start with the following lemma. 
\begin{lemma} 
    \label{ob}
    For any $a \in [n-2]$ and $i \in [n]$, 
    $$\frac{1}{a \cdot {n-1 \choose a}} \sum_{A: |A|=a} v_i(\mathbf{s}_A, \mathbf{0}_{B}, s_i) \geq \frac{1}{(a+1) \cdot {n-1 \choose a+1}} \sum_{A: |A|=a+1} v_i(\mathbf{s}_A, \mathbf{0}_{B}, s_i).$$
\end{lemma}
\begin{proof}
Before the formal proof, let us give some intuition. 
$\frac{1}{ {n-1 \choose a}} \sum_{A: |A|=a} v_i(\mathbf{s}_A, \mathbf{0}_{B}, s_i)$ is the expected value of $v_i(\mathbf{s}_A, \mathbf{0}_{B}, s_i)$ when $A$ is an uniform random set of size $a$. So, $\frac{1}{a \cdot {n-1 \choose a}} \sum_{A: |A|=a} v_i(\mathbf{s}_A, \mathbf{0}_{B}, s_i)$ can be viewed as an amortized expected value. This lemma says that this amortized expected value decreases in terms of the set size $a$. This is not surprising given the SOS property of the valuation function. 

Now we prove it formally. After canceling common factors in the binomial coefficients of both sides, the inequality is equivalent to the following one 
    $$ \sum_{A: |A|=a} v_i(\mathbf{s}_A, \mathbf{0}_{B}, s_i) \geq \frac{a}{n-1-a} \sum_{A: |A|=a+1} v_i(\mathbf{s}_A, \mathbf{0}_{B}, s_i).$$
We shall prove this inequality in the remaining of the proof. First we have
\begin{equation}\label{equ:1}
\sum_{A: |A| = a} v_i(\mathbf{s}_A, \mathbf{0}_B, s_i)
= \frac{1}{n-1-a}  \sum_{A: |A| = a+1} \sum_{j \in A} v_i(\mathbf{s}_{A \setminus \{j\}},\mathbf{0}_{B \cup \{j\}}, s_i).
\end{equation}
This identity holds since each term in the LHS is counted $n-1-a$ times in RHS. Every set of size $a$ can be extended to $n-1-a$ different sets of size $a+1$. 

By applying Lemma \ref{lemma:prev-value-deviation}, we get
\[\sum_{j \in A}(v_i(\mathbf{s}_{A}, \mathbf{0}_{B}, s_i)-  v_i(\mathbf{s}_{A \setminus \{j\}}, \mathbf{0}_{B \cup \{j\}}, s_i)) \leq v_i(\mathbf{s}_{A}, \mathbf{0}_{B}, s_i). \]
After rearranging the terms we get
\begin{equation}\label{equ:2}
  \sum_{j \in A} v_i(\mathbf{s}_{A \setminus \{j\}}, \mathbf{0}_{B \cup \{j\}}, s_i) \geq (|A|-1)v_i(\mathbf{s}_A, \mathbf{0}_B, s_i).
\end{equation}
Connecting \eqref{equ:1} and \eqref{equ:2}, we get
\[\sum_{A: |A| = a} v_i(\mathbf{s}_A, \mathbf{0}_B, s_i)
\geq \frac{a}{n-1-a}  \sum_{A: |A| = a+1} v_i(\mathbf{s}_A, \mathbf{0}_B, s_i)\]
This concludes the proof.
\end{proof}

We can keep applying this monotonicity lemma and bound all these summations by the value of $v_i(\mathbf{s})$. 

\begin{corollary}
For any $0 \leq a \leq n - 1$
$$\sum_{A: |A| = a} v_i(\mathbf{s}_A, \mathbf{0}_B, s_i) \geq \frac{a \cdot {n-1 \choose a}}{n-1} v_i(\mathbf{s})$$
\end{corollary}
\begin{proof}

This inequality is trivial when $a=n-1$. For $a\in [n-2]$, we can keep using the above lemma to get the proof.  

\begin{align*}
    & \frac{1}{a \cdot {n-1 \choose a}} \sum_{A: |A| = a} v_i(\mathbf{s}_A, \mathbf{0}_B, s_i) \\
    \geq& \frac{1}{(a+1) \cdot {n-1 \choose a+1}} \sum_{A: |A| = a+1} v_i(\mathbf{s}_A, \mathbf{0}_B, s_i) \\
    \geq& \cdots \\
    \geq& \frac{1}{(n-1) \cdot {n-1 \choose n-1}} \sum_{A: |A| = n-1} v_i(\mathbf{s}_A, \mathbf{0}_B, s_i) \\
    =& \frac{1}{n-1} v_i(\mathbf{s}).
\end{align*}
\end{proof}

\begin{lemma}
\label{lemma:key_lemma}
     Let $A$ be a random subset of $[n] \setminus \{i\}$, where each bidder in $A$ is chosen with probability $p$, and let $B := ([n] \setminus \{i\}) \setminus A$. For any $\mathbf{s}$,
    $$\mathbb{E}_A[v_i(\mathbf{s}_A, \mathbf{0}_B, s_i)] \geq p \cdot v_i(\mathbf{s})$$
\end{lemma}

\begin{proof}
\begin{align*}
    &\mathbb{E}_A[v_i(\mathbf{s}_A, \mathbf{0}_B, s_i)] \\
    =& \sum_{A} p^{|A|} (1-p)^{|B|} \cdot v_i(\mathbf{s}_A, \mathbf{0}_B, s_i) \\
    =& \sum_{a=1}^{n-1} p^{a} (1-p)^{n-1-a} \cdot \sum_{A: |A|=a} v_i(\mathbf{s}_A, \mathbf{0}_B, s_i) \\
    \geq& \sum_{a=1}^{n-1} p^{a} (1-p)^{n-1-a} \cdot \frac{a}{n-1} {n-1 \choose a} v_i(\mathbf{s}) \\
    =& p \cdot v_i(\mathbf{s}) \cdot \sum_{a=1}^{n-1} p^{a-1} (1-p)^{n-1-a} {n-2 \choose a-1} \\
    =& p \cdot v_i(\mathbf{s}) \cdot (p + (1 - p))^{n-2} \\
    =& p \cdot v_i(\mathbf{s})
\end{align*}
\end{proof}

\noindent \textit{Proof of Theorem \ref{thm:random}.}
Without loss of generality, we assume that the agent 1 and 2 achieve the largest and the second largest values respectively given the signal profile $\mathbf{s}$. We calculate the social welfare of the auction from two disjoint events $1\in B$ and 
$2 \in B \land 1 \in A$. 


\begin{align*}
    &\mathbb{E} \left[ \max_{i \in B} w_i \right] \\
    \geq &\mathbb{E} \left[ \max_{i \in B} w_i \cdot \mathbf{1}_{1 \in B} \right] +\mathbb{E} \left[ \max_{i \in B} w_i \cdot \mathbf{1}_{2 \in B \land 1 \in A} \right]  \\
    =& \mathbb{E} \left[ w_1 \cdot \mathbf{1}_{1 \in B} \right] + \mathbb{E} \left[ w_2 \cdot \mathbf{1}_{2 \in B \land 1 \in A} \right]\\
    =& \mathbb{E} \left[ v_1(\mathbf{s}_1,\mathbf{s}_A,\mathbf{0}_{B_{-1}}) \mid 1 \in B \right]  \cdot \pr(1 \in B) + \mathbb{E} \left[ v_2(\mathbf{s}_2,\mathbf{s}_A,\mathbf{0}_{B_{-2}}) \mid 2 \in B \land 1 \in A \right]  \cdot \pr(2 \in B \land 1 \in A) \\
    \geq& p \cdot v_1(\mathbf{s}) \cdot (1-p) + p \cdot v_2(\mathbf{s}) \cdot p(1-p),
\end{align*}
where the last inequality uses Lemma \ref{lemma:key_lemma}. 

Since the optimal social welfare is $v_{(1)}(\mathbf{s})$, the approximation ratio of the random sampling mechanism is at least 
    $$p(1-p)\left(1 + p \cdot \frac{v_{(2)}(\mathbf{s})}{v_{(1)}(\mathbf{s})}\right).$$
\qed

\section{Mechanism}

\begin{theorem}
\label{thm:main}
For agents with SOS valuations, there is a polynomial time, ex-post IC \& IR  mechanism that gives 3.31543-approximation to the optimal welfare.
\end{theorem}

\begin{proof}
The final mechanism is a convex combination of the above two mechanisms: run the contribution-base mechanism with probability $q$ and the random sampling mechanism with probability $1-q$. We note that the sampling probability $p$ within the random sampling mechanism and this combination probability $q$ are two parameters of the mechanism  to be fixed later.  

It is obvious that this mechanism is polynomial time and ex-post IC \& IR since both contribution base mechanism and random sampling mechanism are. 

The approximation ratio is just the convex combination of the two mechanisms. 
\begin{align*}
& \frac{1}{2} \left( 1-\frac{v_{(2)}(\mathbf{s})}{v_{(1)}(\mathbf{s})} \right)q + p(1-p)\left(1 + p \cdot \frac{v_{(2)}(\mathbf{s})}{v_{(1)}(\mathbf{s})}\right) (1-q) \\
=& \frac{q}{2} + p(1-p)(1-q) +  \left(p^2(1-p)(1-q)-\frac{q}{2}\right) \frac{v_{(2)}(\mathbf{s})}{v_{(1)}(\mathbf{s})}.
\end{align*}

By choosing 
$q= \frac{2p^2(1-p)}{1+2 p^2(1-p)}$,
The coefficient of  $\frac{v_{(2)}(\mathbf{s})}{v_{(1)}(\mathbf{s})}$ vanishes since
$p^2(1-p)(1-q)-\frac{q}{2}=0$. The approximation ratio is then 
\[\frac{q}{2} + p(1-p)(1-q) =\frac{p(1-p^2)}{1+2 p^2(1-p)}. \]

 Let $p$ be the non-negative real solution of $2x^4-4x^3+5x^2-1=0$ ($p \approx 0.54056$  and thus $q \approx 0.21167$), 
we get the final approximation ratio $0.30162=\frac{1}{3.31543}$.
This concludes the proof of our main result. 

\end{proof}

\section{Strong-SOS}
In \cite{Eden2019SOS}, a stronger notion called Strong-SOS was also proposed.  

\begin{definition}[Strong-SOS]
     A valuation function $v(\mathbf{s})$ is \emph{strong submodular over signals} if for all bidders $i$, $\mathbf{s}'\succeq \mathbf{s}$ and $\delta\geq 0$,
     $$v(\mathbf{s}'_{-i}, s'_i+\delta) - v(\mathbf{s}'_{-i}, s'_i) \leq v(\mathbf{s}_{-i}, s_i+\delta) - v(\mathbf{s}_{-i}, s_i).$$
\end{definition}

Although these two definitions look similar and seem to only differ in a small technical condition, we shall argue that the concept of SOS is much more natural and robust than strong-SOS.  Signal is an abstract of some private information of the agents and the number represents the strength of the signal. In many cases, only the relative order of different signals rather than their concrete numbers matter since different representations of the signal may have complete different numbers. In particular, if there is a monotone mapping $\phi_i: \mathcal{S}_i\rightarrow \mathcal{S}'_i$ to change one representation of the signals to another, this should not change the problem at all since a mechanism for one representation can be directly transformed to the other with exactly the same performance and behavior. The property of SOS is also invariant for different representations: the valuation is SOS before the monotone mapping iff it is SOS after the mapping. This is desirable and shows the robustness of the definition.  
However, this invariant does not hold for the definition of strong-SOS. As a result, strong-SOS is not a property for the valuation function alone but a property for valuation function combined with a particular representation of the signal space.  

Another advantage for the definition of SOS is that it does not require any additional structure or property in the signal space other than the ordering structure. For strong-SOS, it requires an additional metric structure so that we can define addition. Furthermore, it requires the space to be continuous such as an interval of real numbers or integer numbers, otherwise it may trivialize the definition.  For example, if the space contains four numbers $\mathcal{S}_i=\{0,1,3,7\}$, there does not exists any $s'_i\neq s_i\in \mathcal{S}_i$ and $\delta\neq 0$ such that both $s'_i+\delta$ and $s_i+\delta$ are in  $\mathcal{S}_i$. When  
$\delta=0$, the condition in the strong-SOS property is trivial; when $s'_i= s_i$, the property degenerates to SOS. 

The above observation says that the property of strong-SOS crucially depends on the property of the signal space. So, it may not be that robust and widely applicable. In the following, we argue that it is not that special either. The informal statement is that for any SOS valuation, there exists a monotone mapping of the signal space such that it becomes strong-SOS after the mapping. The take away here is that one may abandon the concept of strong-SOS and focus mainly on SOS. On the other hand, we can also interpret it positively: one can make use of the strong-SOS property freely when designing mechanisms if it is helpful. Then the mechanism can be transformed to general SOS functions.     

\begin{theorem}
    If there exists a mechanism with $\alpha$-approximation for strong-SOS valuations, then there exists one with $(1-O(\epsilon))  \alpha$-approximation for general SOS valuations with finite discrete signal spaces. 
\end{theorem}

We only prove the result for finite discrete signal spaces for simplicity. We believe that it also holds for continuous space (maybe under some smoothness condition for the valuation functions such as Lipschitz condition).   
The detailed formal proof below is not that informative since most of the technical effort is to deal with the oddness for the definition of strong-SOS. The high level idea is simple: just find a mapping. As long as this mapping grows very fast (so we choose  exponential functions here), it becomes strong-SOS.   
However, the space is not continuous after the mapping. We need to fill the holes, we use convex combination to fill the holes. 

\begin{proof}
    Assume that in the auction, each bidder's valuation function $v_i$ is SOS. First, we convert the SOS valuation functions $\{v_i\}_{i \in [n]}$ into a new set of strong-SOS valuation functions $\{\bar{v}_i\}_{i \in [n]}$. Based on the results of $\bar{\mathbf{x}}$ on $\{\bar{v}_i\}_{i \in [n]}$, we construct a new mechanism $\mathcal{M}$ that achieves $\alpha(1-O(\epsilon))$-approximation for the original SOS valuation functions $\{v_i(\mathbf{s})\}_{i \in [n]}$. 
    
    We first show the construction for $\{ \bar{v}_i \}_{i \in [n]}$.
    For finite discrete signal spaces $\mathcal{S}_i$, it is without loss of generality to assume that they are simply consecutive integers starting from zero.  
    Our idea is to extend the original signal space of each bidder $\mathcal{S}_i$ to $\bar{\mathcal{S}}_i = [\sum_{k=0}^{|\mathcal{S}_i|-1} c^k]$, where $c = \lceil \frac{\max_{\mathbf{s} \in \mathcal{S}} v_{(1)}(\mathbf{s})}{\epsilon} \rceil + 1$ correlates to the maximum valuation of a single bidder over all signals. The signal space of all bidders is thus $\bar{\mathcal{S}} = \prod_{i \in [n]} \bar{\mathcal{S}}_i$.
    Specifically, each signal $s_i \in \mathcal{S}_i$ is mapped to an exponential signal $c_{s_i} = \sum_{k = 1}^{s_i - 1} c^k \in 
    \bar{\mathcal{S}}$. We define the set $\{c_s = \sum_{k=1}^s c^k \}_{s \in [\max_{i \in [n]} |\mathcal{S}_i|]}$ as $\mathcal{C}$. Signal $\bar{s}_i \notin \mathcal{C}$ is a convex combination of its two closest signals in $\mathcal{C}$. Formally, for each signal $\bar{s}_i \in \bar{\mathcal{S}}_i$, let $\ell(\bar{s}_i) = \max_{s_i \in \mathcal{S}, c_{s_i} \leq \bar{s}_i} s_i$ and $\mathscr{r}(\bar{s}_i) = \min_{s_i \in \mathcal{S}, c_{s_i} \geq \bar{s}_i} s_i$ be the two closest integers smaller and larger than $\bar{s}_i$ separately.
    Notice that $\ell(\bar{s}_i) = \mathscr{r}(\bar{s}_i)$ when $\bar{s}_i \in \mathcal{C}$, and $\ell(\bar{s}_i) + 1 = \mathscr{r}(\bar{s}_i)$ otherwise. For the convenience of notation, we define this convex decomposition $\mu(\bar{s}_i): \bar{\mathcal{S}}_i \rightarrow \Delta_{\mathcal{S}_i}$ as a distribution where $P(\ell(\bar{s}_i))=\frac{c_{\ell(\bar{s}_i)+1} - \bar{s}_i}{c_{\ell(\bar{s}_i)+1} - c_{\ell(\bar{s}_i)}}$ and $P(\ell(\bar{s}_i)+1) = \frac{\bar{s}_i - c_{\ell(\bar{s}_i)}}{c_{\ell(\bar{s}_i)+1} - c_{\ell(\bar{s}_i)}}$. 
    The decomposition of a signal profile $\mu(\bar{\mathbf{s}})$ is the joint distribution over the decomposition of each bidder's signal $\bar{s}_i$, i.e., $\mu(\bar{\mathbf{s}}) = \prod_{i = 1}^n \mu(\bar{s}_i)$.
    
    We construct valuation function $\bar{v}_i(\bar{\mathbf{s}})$ as the expectation over $\mu(\bar{\mathbf{s}})$ plus a small number, defined as the following
    $$\bar{v}_i(\bar{\mathbf{s}}) = \mathbb{E}_{\mathbf{s} \sim \mu(\bar{\mathbf{s}})} \left[ v_i(\mathbf{s}) + \epsilon \cdot \lVert \mathbf{s} \rVert_1 \right]$$. 

\begin{lemma}
    For any $i \in [n]$, the constructed valuation function $\bar{v}_i: \bar{\mathcal{S}} \rightarrow \mathbb{R}$ is strong-SOS.
\end{lemma}

\begin{proof}

    We will first show that $\bar{v}_i$ is an SOS valuation function. 
    Next, we will show that for any $\bar{\mathbf{s}}_{-j}$, $\bar{v}_i(\bar{\mathbf{s}}_{-j}, \bar{s}_j)$ is a convex sequence in $\bar{s}_j$, i.e., $\bar{v}_i(\bar{\mathbf{s}}_{-j}, \bar{s}_j + 1) - \bar{v}_i(\bar{\mathbf{s}}_{-j}, \bar{s}_j) \leq \bar{v}_i(\bar{\mathbf{s}}_{-j}, \bar{s}_j) - \bar{v}_i(\bar{\mathbf{s}}_{-j}, \bar{s}_j - 1)$. This implies that for any $s'_j \leq s_j$ and $\delta \geq 0$, $\bar{v}_i(\bar{\mathbf{s}}_{-j}, \bar{s}_j + \delta) - \bar{v}_i(\bar{\mathbf{s}}_{-j}, \bar{s}_j) \leq \bar{v}_i(\bar{\mathbf{s}}_{-j}, \bar{s}'_j + \delta) - \bar{v}_i(\bar{\mathbf{s}}_{-j}, \bar{s}'_j)$. By combining these two results, we can conclude that $\bar{v}_i$ is strong-SOS.
    
    We start by proving a useful claim.
    
    \begin{claim} \label{claim:strong-sos}
    For any $\ell \in [\mathcal{S}-1]$ and $\bar{s}'_j \leq \bar{s}_j$,
    $$\bar{v}_i(\bar{\mathbf{s}}_{-jk}, \bar{s}_j, c_{\ell+1}) - \bar{v}_i(\bar{\mathbf{s}}_{-jk}, \bar{s}_j, c_{\ell}) \leq \bar{v}_i(\bar{\mathbf{s}}_{-jk}, \bar{s}'_j, c_{\ell+1}) - \bar{v}_i(\bar{\mathbf{s}}_{-jk}, \bar{s}'_j, c_{\ell})$$
    \end{claim}

    \begin{proof}
    If $r(\bar{s}_j) - 1 = \ell(\bar{s}'_j)$,
    
    \begin{align*}
        & \left( \bar{v}_i(\bar{\mathbf{s}}_{-jk}, \bar{s}_j, c_{\ell+1}) - \bar{v}_i(\bar{\mathbf{s}}_{-jk}, \bar{s}_j, c_{\ell}) \right) - \left( \bar{v}_i(\bar{\mathbf{s}}_{-jk}, \bar{s}'_j, c_{\ell+1}) - \bar{v}_i(\bar{\mathbf{s}}_{-jk}, \bar{s}'_j, c_{\ell}) \right) \\
        =& \frac{\bar{s}_j - \bar{s}'_j}{c_{\ell(\bar{s}^\prime_j) + 1} - c_{\ell(\bar{s}^\prime_j)}} \left( \bar{v}_i(\bar{\mathbf{s}}_{-jk}, c_{\ell(\bar{s}_j)+1}, c_{\ell+1}) - \bar{v}_i(\bar{\mathbf{s}}_{-jk}, c_{\ell(\bar{s}_j)+1}, c_{\ell}) \right) + \\
        & \frac{\bar{s}'_j - \bar{s}_j}{c_{\ell(\bar{s}^\prime_j) + 1} - c_{\ell(\bar{s}^\prime_j)}} \left( \bar{v}_i(\bar{\mathbf{s}}_{-jk}, c_{\ell(\bar{s}_j)}, c_{\ell+1}) - \bar{v}_i(\bar{\mathbf{s}}_{-jk}, c_{\ell(\bar{s}_j)}, c_{\ell}) \right) \\
        =& \frac{\bar{s}_j - \bar{s}'_j}{c_{\ell(\bar{s}^\prime_j) + 1} - c_{\ell(\bar{s}^\prime_j)}} \left[ \left( \bar{v}_i(\bar{\mathbf{s}}_{-jk}, c_{\ell(\bar{s}_j)+1}, c_{\ell+1}) - \bar{v}_i(\bar{\mathbf{s}}_{-jk}, c_{\ell(\bar{s}_j)+1}, c_{\ell}) \right) - \right. \\
        & \quad \quad \quad \quad \quad \quad \; \, \, \left. \left( \bar{v}_i(\bar{\mathbf{s}}_{-jk}, c_{\ell(\bar{s}_j)}, c_{\ell+1}) - \bar{v}_i(\bar{\mathbf{s}}_{-jk}, c_{\ell(\bar{s}_j)}, c_{\ell}) \right) \right] \\
        =& \frac{\bar{s}_j - \bar{s}'_j}{c_{\ell(\bar{s}^\prime_j) + 1} - c_{\ell(\bar{s}^\prime_j)}} \mathbb{E}_{\mathbf{s}_{-jk} \sim \mu(\bar{\mathbf{s}}_{-jk})} \left[ \left( v_i(\mathbf{s}_{-jk}, \ell(\bar{s}_j)+1, \ell+1) - v_i(\mathbf{s}_{-jk}, \ell(\bar{s}_j)+1, \ell) + \epsilon \right) - \right. \\
        & \quad \quad \quad \quad \quad \quad \quad \quad \quad \quad \quad \; \left. \left( v_i(\mathbf{s}_{-jk}, \ell(\bar{s}_j), \ell+1) - v_i(\mathbf{s}_{-jk}, \ell(\bar{s}_j), \ell) + \epsilon \right) \right] \\
        \leq& 0
    \end{align*}
    
    $\frac{\bar{s}_j - \bar{s}'_j}{c_{\ell(\bar{s}_j) + 1} - c_{\ell(\bar{s}_j)}} \geq 0$ since $\bar{s}'_j \leq \bar{s}_j$. For each $\mathbf{s}_{-jk} \sim \mu(\bar{\mathbf{s}}_{-jk})$, since $v_i$ is SOS, $v_i(\mathbf{s}_{-jk}, \ell(\bar{s}_j)+1, \ell+1) - v_i(\mathbf{s}_{-jk}, \ell(\bar{s}_j)+1, \ell) \leq v_i(\mathbf{s}_{-jk}, \ell(\bar{s}_j), \ell+1) - v_i(\mathbf{s}_{-jk}, \ell(\bar{s}_j), \ell)$. By linearity of expectation, the second part of equation is negative.

    If $r(\bar{s}_j) - 1 > \ell(\bar{s}'_j)$,
    \begin{align*}
        &\bar{v}_i(\bar{\mathbf{s}}_{-jk}, \bar{s}_j, c_{\ell+1}) - \bar{v}_i(\bar{\mathbf{s}}_{-jk}, \bar{s}_j, c_{\ell}) \\
        \leq& \bar{v}_i(\bar{\mathbf{s}}_{-jk}, c_{\ell(\bar{s}_j)}, c_{\ell+1}) - \bar{v}_i(\bar{\mathbf{s}}_{-jk}, c_{\ell(\bar{s}_j)}, c_{\ell}) \\
        \leq & \bar{v}_i(\bar{\mathbf{s}}_{-jk}, c_{\ell(\bar{s}_j) - 1}, c_{\ell+1}) - \bar{v}_i(\bar{\mathbf{s}}_{-jk}, c_{\ell(\bar{s}_j) - 1}, c_{\ell})\\
        \leq& ...\\
        \leq& \bar{v}_i(\bar{\mathbf{s}}_{-jk}, c_{\ell(\bar{s}^\prime_j)+1}, c_{\ell+1}) - \bar{v}_i(\bar{\mathbf{s}}_{-jk}, c_{\ell(\bar{s}^\prime_j)+1}, c_{\ell})\\
        \leq& \bar{v}_i(\bar{\mathbf{s}}_{-jk}, \bar{s}'_{j}, c_{\ell+1}) - \bar{v}_i(\bar{\mathbf{s}}_{-jk}, , \bar{s}'_{j}, c_{\ell})
    \end{align*}
    \end{proof}
    
    With the result from Claim \ref{claim:strong-sos}, we are ready to show that $\bar{v}_i$ is SOS.
    
    \begin{lemma}
        For any $\bar{\mathbf{s}}_{-j} \in \bar{\mathcal{S}}_{-j}$,  $s^\prime_j \leq s_j$ and $\delta \geq 0$, 
        $$\bar{v}_i(\bar{\mathbf{s}}_{-j}, \bar{s}_j + \delta) -  \bar{v}_i(\bar{\mathbf{s}}_{-j}, \bar{s}_j) \leq \bar{v}_i(\bar{\mathbf{s}}^\prime_{-j}, \bar{s}_j +\delta) - \bar{v}_i(\bar{\mathbf{s}}^\prime_{-j}, \bar{s}_j)$$
    \end{lemma}
    
    \begin{proof}
    
        We will first show that $\bar{v}_i$ is SOS when $\ell(\bar{s}_i) = \ell(\bar{s}_i + \delta)$, and then prove the more generalized result. Specifically, let us define $\ell = \ell(\bar{s}_i) = \ell(\bar{s}_i + \delta)$.

        \begin{align*}
        &\bar{v}_i(\bar{\mathbf{s}}_{-j}, \bar{s}_j + \delta) -  \bar{v}_i(\bar{\mathbf{s}}_{-j}, \bar{s}_{j}) \\
        =& \frac{\bar{s}_j + \delta - c_{\ell}}{c_{\ell+1} - c_{\ell}} \bar{v}_i(\bar{\mathbf{s}}_{-j}, c_{\ell+1}) + \frac{c_{\ell+1} - (\bar{s}_j + \delta)}{c_{\ell+1} - c_{\ell}} \bar{v}_i(\bar{\mathbf{s}}_{-j}, c_{\ell}) - \frac{\bar{s}_j - c_{\ell}}{c_{\ell+1} - c_{\ell}} \bar{v}_i(\bar{\mathbf{s}}_{-j}, c_{\ell+1}) - \frac{c_{\ell+1} - \bar{s}_j}{c_{\ell+1} - c_{\ell}} \bar{v}_i(\bar{\mathbf{s}}_{-j}, c_{\ell}) \\
        =& \frac{\delta}{c_{\ell+1} - c_{\ell}} \left( \bar{v}_i(\bar{\mathbf{s}}_{-j},c_{\ell+1}) -\bar{v}_i(\bar{\mathbf{s}}_{-j}, c_{\ell}) \right)\\
        \leq& \frac{\delta}{c_{\ell+1} - c_{\ell}} \left( \bar{v}_i(\bar{\mathbf{s}}_{-\{1,j\}}, \bar{s}'_1, c_{\ell+1}) -\bar{v}_i(\bar{\mathbf{s}}_{-\{1,j\}}, \bar{s}'_1, c_{\ell}) \right) \\
        \leq& \frac{\delta}{c_{\ell+1} - c_{\ell}} \left( \bar{v}_i(\bar{\mathbf{s}}_{-\{1,2, j\}}, \bar{s}'_1, \bar{s}'_2, c_{\ell+1}) -\bar{v}_i(\bar{\mathbf{s}}_{-\{1,2, j\}}, \bar{s}'_1, \bar{s}'_2, c_{\ell}) \right) \\
        \leq& \cdots \\
        \leq& \frac{\delta}{c_{\ell+1} - c_{\ell}} \left( \bar{v}_i(\bar{\mathbf{s}}'_{-j},c_{\ell+1}) -\bar{v}_i(\bar{\mathbf{s}}'_{-j}, c_{\ell}) \right) \\
        =& \frac{\bar{s}_j + \delta - c_{\ell}}{c_{\ell+1} - c_{\ell}} \bar{v}_i(\bar{\mathbf{s}}'_{-j}, c_{\ell+1}) + \frac{c_{\ell+1} - (\bar{s}_j + \delta)}{c_{\ell+1} - c_{\ell}} \bar{v}_i(\bar{\mathbf{s}}'_{-j}, c_{\ell}) - \frac{\bar{s}_j - c_{\ell}}{c_{\ell+1} - c_{\ell}} \bar{v}_i(\bar{\mathbf{s}}'_{-j}, c_{\ell+1}) - \frac{c_{\ell+1} - \bar{s}_j}{c_{\ell+1} - c_{\ell}} \bar{v}_i(\bar{\mathbf{s}}'_{-j}, c_{\ell}) \\
        =& \bar{v}_i(\bar{\mathbf{s}}^\prime_{-j}, \bar{s}_j +\delta) - \bar{v}_i(\bar{\mathbf{s}}^\prime_{-j}, \bar{s}_j)
    \end{align*}
    The inequality above is from Claim \ref{claim:strong-sos}.
    
    We will next show that $\bar{v}_i$ is SOS when $\ell(\bar{s}_j)$ and $\ell(\bar{s}_j + \delta)$ could possibly be different, i.e., $\ell(\bar{s}_j) \leq \ell(\bar{s}_j + \delta)$.
    \begin{align*}
        &\bar{v}_i(\bar{\mathbf{s}}_{-j}, \bar{s}_j + \delta) -  \bar{v}_i(\bar{\mathbf{s}}_{-j},\bar{s}_{j})\\
        =& \left(\bar{v}_i(\bar{\mathbf{s}}_{-j}, c_{\ell(\bar{s}_j)+1}) - \bar{v}_i(\bar{\mathbf{s}}_{-j}, \bar{s}_{j}) \right) + \left(\bar{v}_i(\bar{\mathbf{s}}_{-j}, c_{\ell(\bar{s}_j)+2}) - \bar{v}_i(\bar{\mathbf{s}}_{-j}, c_{\ell(\bar{s}_j)+1}) \right) + \cdots + \\
        &\left(\bar{v}_i(\bar{\mathbf{s}}_{-j}, \bar{s}_j + \delta) - \bar{v}_i(\bar{\mathbf{s}}_{-j}, c_{\ell(\bar{s}_j + \delta)}) \right) \\
        \leq& \left(\bar{v}_i(\bar{\mathbf{s}}'_{-j}, c_{\ell(\bar{s}_j)+1}) - \bar{v}_i(\bar{\mathbf{s}}_{-j}, \bar{s}_{j}) \right) + \left(\bar{v}_i(\bar{\mathbf{s}}'_{-j}, c_{\ell(\bar{s}_j)+2}) - \bar{v}_i(\bar{\mathbf{s}}'_{-j}, c_{\ell(\bar{s}_j)+1}) \right) + \cdots + \\
        &\left(\bar{v}_i(\bar{\mathbf{s}}'_{-j}, \bar{s}_j + \delta) - \bar{v}_i(\bar{\mathbf{s}}'_{-j}, c_{\ell(\bar{s}_j + \delta)}) \right) \\
        =& \bar{v}_i(\bar{\mathbf{s}}'_{-j}, \bar{s}_j + \delta) - \bar{v}_i(\bar{\mathbf{s}}'_{-j}, \bar{s}_j)
    \end{align*}
    \end{proof}
    
    Next, we will show that for any signal profile $\bar{\mathbf{s}}$, $\bar{v}_i(\bar{\mathbf{s}}_{-j}, \bar{s}_j + 1) - \bar{v}_i(\bar{\mathbf{s}}_{-j}, \bar{s}_j) \leq \bar{v}_i(\bar{\mathbf{s}}_{-j}, \bar{s}_j) - \bar{v}_i(\bar{\mathbf{s}}_{-j}, \bar{s}_j - 1)$.
    
    When $\ell(\bar{s}_j - 1) = \ell(\bar{s}_j) = \ell$, $\bar{v}_i(\bar{\mathbf{s}}_{-j}, \bar{s}_j + 1) - \bar{v}_i(\bar{\mathbf{s}}_{-j}, \bar{s}_j)$ and $\bar{v}_i(\bar{\mathbf{s}}_{-j}, \bar{s}_j) - \bar{v}_i(\bar{\mathbf{s}}_{-j}, \bar{s}_j - 1)$ both can be rewritten as $\frac{1}{c_{\ell + 1} - c_{\ell}} \left( \bar{v}_i(\bar{\mathbf{s}}_{-j}, c_{\ell + 1}) -  \bar{v}_i(\bar{\mathbf{s}}_{-j}, c_{\ell}) \right)$ and thus, $$\bar{v}_i(\bar{\mathbf{s}}_{-j}, \bar{s}_j + 1) - \bar{v}_i(\bar{\mathbf{s}}_{-j}, \bar{s}_j) = \bar{v}_i(\bar{\mathbf{s}}_{-j}, \bar{s}_j) - \bar{v}_i(\bar{\mathbf{s}}_{-j}, \bar{s}_j - 1)$$
    
    Otherwise, $\ell(\bar{s}_j - 1) < \ell(\bar{s}_j)$ implies that $\ell(\bar{s}_j - 1) = \ell(\bar{s}_j) - 1$. $\bar{s}_j \in \mathcal{C}$. 
    
    \begin{align*}
        &\bar{v}_i(\bar{\mathbf{s}}_{-j}, \bar{s}_j + 1) - \bar{v}_i(\bar{\mathbf{s}}_{-j}, \bar{s}_j) \\
        =& \frac{1}{c_{\ell(\bar{s}_j) + 1} - c_{\ell(\bar{s}_j)}} \left( \bar{v}_i(\bar{\mathbf{s}}_{-j}, c_{\ell(\bar{s}_j)+1}) - \bar{v}_i(\bar{\mathbf{s}}_{-j}, c_{\ell(\bar{s}_j)}) + \epsilon \right) \\
        =& \frac{1}{c^{\ell(\bar{s}_j)}} \left( \mathbb{E}_{\mathbf{s}_{-j} \sim \mu(\bar{\mathbf{s}}_{-j})} \left( v_i(\mathbf{s}_{-j}, \ell(\bar{s}_j)+1) - v_i(\mathbf{s}_{-j}, \ell(\bar{s}_j)) \right) + \epsilon \right) \\
        \leq&  \frac{1}{c^{\ell(\bar{s}_j)}} \left( \max_{s \in \mathcal{S}} v_{(1)}(\mathbf{s}) + \epsilon \right) \\
        \leq& \frac{\epsilon}{c^{\ell(\bar{s}_j)-1}} \\
        \leq& \frac{1}{c^{\ell(\bar{s}_j)-1}} \left( \mathbb{E}_{\mathbf{s}_{-j} \sim \mu(\bar{\mathbf{s}}_{-j})} \left( v_i(\mathbf{s}_{-j}, \ell(\bar{s}_j)) - v_i(\mathbf{s}_{-j}, \ell(\bar{s}_j) - 1) \right) + \epsilon \right) \\
        =& \frac{1}{c_{\ell(\bar{s}_j - 1) + 1} - c_{\ell(\bar{s}_j - 1)}} \left( \bar{v}_i(\bar{\mathbf{s}}_{-j}, c_{\ell(\bar{s}_j)}) - \bar{v}_i(\bar{\mathbf{s}}_{-j}, c_{\ell(\bar{s}_j) - 1}) + \epsilon \right) \\
        =& \bar{v}_i(\bar{\mathbf{s}}_{-j}, \bar{s}_j) - \bar{v}_i(\bar{\mathbf{s}}_{-j}, \bar{s}_j - 1)
    \end{align*}

Therefore, $v_i(\mathbf{s}_{-j}, \mathbf{s}_j)$ forms a convex sequence on $\mathbf{s}_j$. We can then conclude that for any $\delta > 0$, 

$$\bar{v}_i(\bar{\mathbf{s}}_{-j}, \bar{s}_j +\delta) - \bar{v}_i(\bar{\mathbf{s}}_{-j}, \bar{s}_{j})
\leq \bar{v}_i(\bar{\mathbf{s}}_{-j}, \bar{s}'_j +\delta) - \bar{v}_i(\bar{\mathbf{s}}_{-j}, \bar{s}'_{j})$$

Combining the two results we conclude our proof.
$$\bar{v}_i(\bar{\mathbf{s}}_{-j}, \bar{s}_j +\delta) - \bar{v}_i(\bar{\mathbf{s}}_{-j}, \bar{s}_{j})
\leq \bar{v}_i(\bar{\mathbf{s}}_{-j}, \bar{s}'_j +\delta) - \bar{v}_i(\bar{\mathbf{s}}'_{-j}, \bar{s}_{j}) \leq \bar{v}_i(\bar{\mathbf{s}}'_{-j}, \bar{s}'_j +\delta) - \bar{v}_i(\bar{\mathbf{s}}'_{-j}, \bar{s}'_{j})$$
\end{proof}


Finally, we show our reduction and prove the approximation result. Suppose mechanism $\bar{\mathcal{M}} = (\bar{\mathbf{x}}, \bar{p})$ is ex-post IC \& IR and achieves $\alpha$-approximation on any strong-SOS valuation function. We will show that there exists a monotone allocation rule $x$ such that $x$ achieves $\alpha(1 - O(\epsilon))$-approximation on any SOS valuation setting. $x(\mathbf{s})$ simply takes the allocation rule of signal $c_{\mathbf{s}} \in \bar{\mathcal{S}}$, i.e., $\mathbf{x}(\mathbf{s}) = \bar{\mathbf{x}}(c_{\mathbf{s}})$. The monotonicity of $\bar{\mathbf{x}}(\bar{\mathbf{s}})$ directly implies that $x(\mathbf{s})$ is monotone: for any bidder $i \in [n]$, signals $\mathbf{s}_{-i}$ and $s_i \geq s'_i$,
$$x_i(\mathbf{s}_{-i}, s_i) = \bar{x}(c_{s_1}, c_{s_2}, \cdots, c_{s_i}, \cdots, c_{s_n}) \geq \bar{x}(c_{s_1}, c_{s_2}, \cdots, c_{s'_i}, \cdots, c_{s_n}) = x_i(\mathbf{s}_{-i},s'_i)$$

Since $\bar{\mathbf{x}}$ achieves $\alpha$-approximation on strong-SOS valuation functions $\{\bar{v}_i\}_{i \in [n]}$, we have 

$$\min_{\mathbf{s} \in \mathcal{S}} \frac{\sum_{i} x_i(\mathbf{s}) \cdot (v_i(\mathbf{s}) + \epsilon \cdot \lVert \mathbf{s} \rVert_1)}{v_{(1)}(\mathbf{s}) + \epsilon \cdot \lVert \mathbf{s} \rVert_1} 
= \min_{\mathbf{s} \in \mathcal{S}} \frac{\sum_{i} \bar{x}_i(c_{\mathbf{s}}) \cdot \bar{v}_i(c_{\mathbf{s}})}{\bar{v}_{(1)}(c_{\mathbf{s}})} 
\geq \min_{\bar{\mathbf{s}} \in \bar{\mathcal{S}}} \frac{\sum_i \bar{x}_i(\bar{\mathbf{s}}) \cdot \bar{v}_i(\bar{\mathbf{s}})}{\bar{v}_{(1)}(\bar{\mathbf{s}})}
\geq \alpha$$

Based on our construction that $\bar{v}_i(\bar{\mathbf{s}}) = \mathbb{E}_{\mathbf{s} \sim \mu(\bar{\mathbf{s}})} \left[ v_i(\mathbf{s}) + \epsilon \cdot \lVert \mathbf{s} \rVert_1 \right]$,
\begin{align*}
    &\min_{\mathbf{s} \in \mathcal{S}} \frac{\sum_{i} x_i(\mathbf{s}) \cdot v_i(\mathbf{s})}{v_{(1)}(\mathbf{s})} \\ 
    \geq& \min_{\mathbf{s} \in \mathcal{S}} \frac{\sum_{i} x_i(\mathbf{s}) \cdot v_i(\mathbf{s})}{v_{(1)}(\mathbf{s}) + \epsilon \cdot \lVert \mathbf{s} \rVert_1} \\
    \geq& \min_{\mathbf{s} \in \mathcal{S}} \frac{\sum_{i} x_i(\mathbf{s}) \cdot (v_i(\mathbf{s}) + \epsilon \cdot \lVert \mathbf{s} \rVert_1)}{v_{(1)}(\mathbf{s}) + \epsilon \cdot \lVert \mathbf{s} \rVert_1} \cdot \min_{\mathbf{s} \in \mathcal{S}} \frac{\sum_{i} x_i(\mathbf{s}) \cdot v_i(\mathbf{s})}{\sum_{i} x_i(\mathbf{s}) \cdot (v_i(\mathbf{s}) + \epsilon \cdot \lVert \mathbf{s} \rVert_1)} \\
    \geq& \alpha \cdot \min_{\mathbf{s} \in \mathcal{S}, i \in [n]} \frac{v_i(\mathbf{s})}{v_i(\mathbf{s}) + \epsilon \cdot \lVert \mathbf{s} \rVert_1} \\
    \geq& \alpha \cdot \frac{1}{1 + \epsilon \cdot \frac{\max_{\mathbf{s} \in \mathcal{S}} \lVert \mathbf{s} \rVert_1}{\min_{\mathbf{s} \in \mathcal{S}} v_{(1)}(\mathbf{s})}} \\
    =& \alpha \cdot \left( 1 - \frac{\max_{\mathbf{s} \in \mathcal{S}}}{\min_{\mathbf{s} \in {\mathcal{S}}} v_{(1)}(\mathbf{s})} \cdot \epsilon \right)
\end{align*}

Therefore, $x$ has $(1-O(\epsilon)) \alpha$-approximation on $\{v_i(\mathbf{s})\}_{i \in [n]}$

\end{proof}


\bibliographystyle{unsrt}  
\bibliography{reference}

\end{document}